\documentclass[preprint, amsmath, amsymb, aps, pre,nofootinbib,showkeys]{revtex4-2}

\usepackage{graphicx}
\usepackage{dcolumn}
\usepackage{bm}

\usepackage{amsmath}
\usepackage{amssymb}
\usepackage{amsthm}
\newcommand{\dd}{{\rm d}}
\newtheorem*{remark*}{\bf Remark}
\newtheorem{lemma}{Lemma}[section]
\newtheorem{assumption}{Assumption}[section]
\newtheorem{theorem}[lemma]{Theorem}
\newtheorem{corollary}[lemma]{Corollary}
\newtheorem{definition}{Definition}[section]
\newcommand{\num}{\text{Num}}
\usepackage{hyperref}

\usepackage{xcolor}
\usepackage{ulem}
\usepackage{cancel} 

\makeatletter
\def\frontmatter@preabstractspace{1cm}
\makeatother

\makeatletter
\pretocmd{\frontmatter@keys@format}
  {\vfill}{}{}
\makeatother

\usepackage{setspace}
\allowdisplaybreaks

\begin{document}

\title{Mori-Zwanzig formalism: An existence proof for weak solutions of the orthogonal dynamics equation}

\author{Christoph Widder}
\email{christoph.widder@physik.uni-freiburg.de}
\affiliation{
 Institute of Physics, Albert-Ludwigs-Universität Freiburg, Hermann-Herder-Straße 3, 79104 Freiburg im Breisgau, Germany
}




\date{\today}

\begin{abstract}
\vspace{5mm}
In classical statistical physics, the Mori-Zwanzig projection operator technique is used to derive generalized Langevin equations for a random variable (observable). Standard derivations implicitly assume the existence of solutions to the so-called orthogonal dynamics equation as well as the validity of the variation of constants formula (Dyson identity). It was pointed out by Givon, Hald and Kupferman that the existence of the orthogonal dynamics is a subtle issue for infinite-rank projections such as Zwanzig's projection [D. Givon, O. H. Hald, R. Kupferman, \textit{Israel Journal of Mathematics}, 145 (221-241), 2005]. The authors proved the existence of weak solutions for Zwanzig's projection in the context of stationary Hamiltonian systems. To date, this is the only existence proof that allows for an infinite-rank projection, whereas the uniqueness and regularity remain open problems. 

In this article, we generalize the existence proof by Givon et al. to nonstationary non-Hamiltonian systems whose time evolution is given by a quasicontraction semigroup. We establish growth bounds as well as the uniqueness for sufficiently regular solutions (if existent). Finally, we apply our results to Zwanzig's projection using the damped harmonic oscillator as an example. 
\end{abstract}

\maketitle
\newpage


\section{Introduction}

The Nakajima-Mori-Zwanzig projection operator formalism originates from the study of transport phenomena and irreversibility in the late 1950s and 1960s \cite{nakajima1958, zwanzig1960, zwanzig1961, mori1965transport}. The main theme is to use projection operators in order to derive closed kinetic equations for a physical quantity of interest. Depending on the context, these projections act on density matrices, phase space distributions or dynamical variables. It was pointed out by Mori \cite{mori1965transport} that the kinetic equation for a dynamical variable resembles a non-Markovian generalization of the Langevin equation for Brownian motion. Equations of this kind are therefore referred to as generalized Langevin equations. Since the 1960s, numerous generalized Langevin equations have been derived using the projection operator technique \cite{kawasaki73, grabert1977, furukawa1979, darve2009, kawai2011, xing2011, meyer2017, vrugt2019, izvekov2019, ayaz2022, vroylandt2022, koch2024, netz2024, hery2024, izvekov2025}. Generalized Langevin equations are used to obtain coarse-grained models of complex systems that allow for non-Markovian memory effects \cite{berkowitz1983, li2015, klippenstein2021introducing}. They are used to describe the dynamics of atomic, polymeric and colloidal systems \cite{snook2006langevin}, to obtain coarse-grained molecular dynamics for crystalline solids \cite{li2010}, and to derive mode coupling approximations of the liquid-glass transition \cite{ahmadirahmat2025, gotze2009complex, hansen2013}. Furthermore, the projection operator technique is used in the context of open quantum systems \cite{breuer2002theory, seegebrecht2025concept} and dynamical density functional theory \cite{espanol2009derivation}. 

Before we outline previous results and the results obtained in this article, let us introduce the generalized Langevin equation and the orthogonal dynamics equation. To this end, we make use of semigroup theory as it provides a rigorous framework for finite-rank projections. For textbook derivations of generalized Langevin equations, we refer to \cite{grabert2006projection, zwanzig2001nonequilibrium, schilling2022coarse}. 

Let $H$ be a real or complex Hilbert space and let $U(t)$ be a strongly continuous semigroup with generator $L$. The time-evolution for a state vector $z\in H$ is expressed by the orbit map $t\to U(t)z$ for all $t\geq 0$. The equation of motion takes the form
\begin{align}
    \frac{\dd }{\dd t}U(t)z &= U(t)Lz \, , \label{equ:EOM}
\end{align}
for all $z\in D(L)$, where $D(L)$ denotes the domain of $L$. 

Now let $P$ and $Q:=1-P$ be orthogonal projections on $H$. If $PL$ is bounded (e.g. if $P$ is a finite-rank projection), the \textit{bounded perturbation theorem} \cite[p.~158, Theorem~1.3]{engel} yields another strongly continuous semigroup $G(t)$ with generator $QL=L-PL$, the so-called \textit{orthogonal dynamics} \footnote{The significance of the bounded perturbation theorem for the projection operator formalism was previously pointed out by Zhu et al. \cite[Appendix A]{zhu18}.}. Symbolically, one often writes
\begin{align*}
    G(t) &= e^{QLt} \, .
\end{align*}
For all $x\in H$, the orbit maps $u(t):=G(t)x$ are the unique mild solutions of the \textit{orthogonal dynamics equation}\footnote{We note that the orthogonal dynamics equation can equivalently be stated as $\frac{\dd}{\dd t} u(t) = QLQ u(t)$ for all initial values $x$ in the range of $Q$.} \cite[p.~146, Definition 6.3, Proposition~6.4]{engel}
\begin{align}
  \begin{cases}
      \frac{\dd}{\dd t} u(t) &= QL u(t) \\
    u(0) &= x 
  \end{cases}  \label{equ:ACP} \, .
\end{align}
In particular, the semigroup $G(t)$ admits the \textit{variation of constants formula} for all $x\in H$ \cite[p.~161, Corollary~1.7]{engel}
\begin{align}
    G(t)x &= U(t)x-\int^t_0 U(t-s)\widehat{PL} G(s)x \, \dd s \, , \label{equ:variation_of_constants}
\end{align}
where $\widehat{PL}$ is the continuous extension of $PL$. In the physics literature, the variation of constants formula (\ref{equ:variation_of_constants}) is commonly referred to as \textit{Dyson decomposition} or \textit{Dyson-Duhamel identity}. Using the variation of constants formula (\ref{equ:variation_of_constants}), the equations of motion (\ref{equ:EOM}) takes the form of a \textit{generalized Langevin equation} \cite{widder_addendum}:
\begin{align*}
     \frac{\dd}{\dd t} U(t)z &= U(t) PL z+ G(t)QLz +  \int^t_0 U(t-s)\widehat{PL} G(s) QL z  \, \dd s \, ,
\end{align*}
for all $z\in D(L)$. Equations of this form constitute a fundamental part of the projection operator formalism.

The well-posedness of the orthogonal dynamics equation (\ref{equ:ACP}) was first studied by Givon et al. in 2005 \cite{givon2005}. In their work, they proved the existence of unique solutions of eq.~(\ref{equ:ACP}) for finite-rank projections \cite[Theorem 4.1]{givon2005} by a reduction to a set of Volterra equations. Furthermore, they proved the existence of weak solutions of eq.~(\ref{equ:ACP}) for an infinite-rank projection (Zwanzig's projection) for stationary Hamiltonian systems \cite[Theorem 5.11]{givon2005} in analogy to Friedrichs' theory of symmetric hyperbolic systems \cite{friedrichs1954}. To the best of our knowledge, the uniqueness and regularity of these solutions remains unresolved to date. 

We note that the projection operator formalism is well understood for the case of finite-rank projections, whereas infinite-rank projections are much more difficult \cite[Sec.~2c]{widder2026}. The reason for this is that the orthogonal dynamics equation cannot be reduced to a finite set of Volterra equations as opposed to the case of finite-rank projections \cite[Theorem 4.1]{givon2005}. Furthermore, for most infinite-rank projections, the orthogonal dynamics $G(t)$ cannot be obtained from a bounded perturbation of the semigroup $U(t)$. An explicit counterexample for Zwanzig's projection can be found in \cite[Sec.~4a]{widder2026}. 

In this article, we generalize the existence proof by Givon et al. \cite[Theorem 5.11]{givon2005}. The main result is the existence proof from \ref{theorem:existence}, which addresses the abstract Cauchy problem (ACP) 
\begin{align}
  \begin{cases}
       u'(t) &= A u(t) +f(t) \\
    u(0) &= g 
  \end{cases}   \, ,\label{eq:cauchy_problem}
\end{align}
where $'$ denotes the weak derivative, $A$ is a densely defined closable operator on $H$, $f$ is a forcing term, and $g$ is the initial value. 
For 
\begin{align*}
    A:=\overline{QL}Q \, ,
\end{align*}
theorem \ref{theorem:existence} guarantees the existence of weak solutions to the orthogonal dynamics equation, provided the operator $A^\dagger$ has a finite numerical abscissa $\omega$, that is
\begin{equation}
    \omega := \sup\{ \Re(z) | z \in \num(A^\dagger) \} < \infty \, ,
\end{equation}
where $\dagger$ denotes the adjoint and $\text{Num}$ denotes the numerical range. This includes the case of Zwanzig's projection for a large class of non-stationary non-Hamiltonian systems. The line of argument presented follows closely the work by Givon et al. \cite{givon2005}, but we start from a rather abstract point of view. 

This article is structured as follows. In sec.~\ref{sec:weak_solutions}, we introduce the notion of weak solutions  of the ACP (\ref{eq:cauchy_problem}) in analogy to Friedrichs' extension. In sec.~\ref{sec:existence}, we establish an energy estimate under the assumption that $A^\dagger$ has a finite numerical abscissa. Given the energy estimate, the existence of weak solutions (Theorem \ref{theorem:existence}) follows from the Riesz representation theorem. In sec.~\ref{sec:uniqueness}, similarly to the energy estimate, an a priori growth bound for sufficiently regular solutions (if existent) is obtained under the assumption that $A$ has a finite numerical abscissa. An immediate consequence is the uniqueness of sufficiently regular solutions (Corollary \ref{corollary:uniqueness}). Within the scope of this work, we cannot provide a proof for the existence of \textit{unique} solutions. The gap between the existence and uniqueness results obtained are discussed in sec.~\ref{sec:outlook}. In sec.~\ref{sec:application}, we apply our results to the projection operator formalism, i.e. we consider the ACP (\ref{eq:cauchy_problem}) associated to $A=\overline{QL}Q$, where $\overline{QL}$ denotes the closure of $QL$. This is possible if $U(t)$ is a strongly continuous quasicontraction semigroup and if $LQ$ as well as $L^\dagger Q$ are densely defined. In sec.~\ref{sec:zwanzig}, we show that for Zwanzig's projection $P$ there exists a dense set of sufficiently smooth test functions in the range of $Q=1-P$. This shows that the assumptions for the existence and uniqueness results are fulfilled for a large class of nonstationary non-Hamiltonian systems. We demonstrate this for the damped harmonic oscillator using the bump function as initial distribution in sec.~\ref{sec:example}. Final remarks and discussions are given in sec.~\ref{sec:conclusion}.

\section{Weak solutions} \label{sec:weak_solutions}

The main result is the existence theory for the ACP (\ref{eq:cauchy_problem}) presented in sec.~\ref{sec:existence}. The growth bounds and the uniqueness of solutions are discussed in secs.~\ref{sec:uniqueness}-\ref{sec:outlook}. First, we introduce the notion of weak solutions.

\begin{assumption}\label{assumption1}
    Let $H$ be a complex Hilbert space with norm $\|x\|:=\sqrt{(x,x)}$, where $(x,y)$ denotes the scalar product. Let $A$ be a densely defined and closable \footnote{We remark that the closability of $A$ assures that $A^\dagger$ densely defined.} linear operator on $H$. 
\end{assumption}

In the following, we make use of functions on the interval $[0,T]$ with values in $H$. The required measure theory and calculus for the Bochner space $L^2(0,T;H)$ and the Sobolev spaces $H^1(0,T;H)$ and $H^1_0(0,T;H)$ can be found e.g. in \cite[Appendix E.5. and Section 5.9.2.]{evans}. 

We consider the ACP (\ref{eq:cauchy_problem}), where $f \in L^2(0,T;H)$ and $g\in H$. Further, we define the operators $E$ with domain $D(E)$ and $E^+$ with domain $D(E^+)$ according to
\begin{align*}
    Eu(t)&:= u'(t) -Au(t) \, , \\
    D(E) &:= \{ u\in H^1(0,T;H) : Au\in L^2(0,T;H)\} \,, \\
    E^+ v(t) &:= -v'(t)-A^\dagger v(t) \, , \\
    D(E^+) &:= \{ u\in H^1(0,T;H) : A^\dagger u\in L^2(0,T;H)\} \, .
\end{align*}
With this, the Cauchy problem (\ref{eq:cauchy_problem}) reads
\begin{align*}
    \begin{cases}
        Eu &= f \\
    u(0) &= g
    \end{cases} \quad .
\end{align*}

In order to define weak solutions, we require the following adjointness formula, cf. \cite[Lemma 5.3]{givon2005}. 

\begin{lemma}\label{lemma:adjointness}
    For all $u \in D(E)$ and $v\in D(E^+)$, we have
\begin{align}
   (Eu,v)_{L^2(0,T;H)} -(u,E^+ v)_{L^2(0,T;H)} &= (u(T),v(T))-(u(0),v(0)) \label{eq:adjointness} \, .
\end{align}
\end{lemma}
\begin{proof}
    The proof is given in appendix \ref{app:integration_by_parts}. 
\end{proof}

This adjointness formula gives rise to a natural notion of weak solutions, provided $D(E^+) \cap H^1_0(0,T;H)$ is dense in $L^2(0,T;H)$. 

\begin{lemma}
The spaces $D(E)\cap H^1_0(0,T;H)$ and $D(E^+) \cap H^1_0(0,T;H)$ are dense in $L^2(0,T;H)$. 
\end{lemma}
\begin{proof}
    The proof is given in appendix \ref{app:dense_subspace}. 
\end{proof}

The following lemma allows us to define weak solutions in the spirit of Friedrichs' extension for symmetric hyperbolic systems \cite{friedrichs1954}, cf. \cite[Lemmata 5.5/5.6, Corollary 5.7]{givon2005}.

\begin{lemma}\label{def:weak_extension}
    For any $u\in L^2(0,T;H)$, there exists at most one pair $f\in L^2(0,T;H)$ and $g\in H$ such that 
    \begin{align}
        (f,v)_{L^2(0,T;H)} -(u,E^+ v)_{L^2(0,T;H)} + (g,v(0))&= 0 \label{eq:weak_extension}
    \end{align}
    for all $v \in D(E^+)$ with $v(T)=0$. If $f$ and $g$ exist, $u$ is called a \textbf{weak solution}, and we define the \textbf{weak extensions} $\tilde{E},\tilde{S}$ of the linear maps $u\to Eu$ and $u\to u(0)$  according to 
    \begin{align*}
        \tilde{E}u &:= f \, , \\
        \tilde{S}u &:= g \, .
    \end{align*}
\end{lemma}
\begin{proof}
    Let $u=0$. Suppose that eq.~(\ref{eq:weak_extension}) holds for all $v\in D(E^+) \cap H^1_0(0,T;H)$. Then,
    \begin{align*}
        (f,v)_{L^2(0,T;H)} &= 0 \, .
    \end{align*}
    Since $D(E^+) \cap H^1_0(0,T;H)$ is dense in $L^2(0,T;H)$, this implies $f=0$. Now suppose that eq.~(\ref{eq:weak_extension}) holds for all $v\in D(E^+)$ with $v(T)=0$. Then we still have $f=0$. Hence,
    $(g,v(0))= 0$ for all $v\in D(E^+)$ with $v(T)=0$. Since $D(A^\dagger)$ is dense in $H$, this implies $g=0$.

    Now let $u\in L^2(0,T;H)$ and suppose there exist functions $f_1,g_1$ and $f_2,g_2$ such that eq.~(\ref{eq:weak_extension}) holds for all $v\in D(E^+)$ with $v(T)=0$. Then
    \begin{align*}
        (f_1-f_2,v)_{L^2(0,T;H)} -(0,E^+ v)_{L^2(0,T;H)} + (g_1-g_2,v(0))&= 0 \, .
    \end{align*}
    Hence, eq.~(\ref{eq:weak_extension}) holds for $u=0$, $f=f_1-f_2$, $g=g_1-g_2$ and for all $v\in D(E^+)$ with $v(T)=0$. By the previous argument, this shows that $f=f_1-f_2=0$ and $g=g_1-g_2=0$. Hence, there exists at most one pair $f$ and $g$ such that eq.~(\ref{eq:weak_extension}) holds for all $v\in D(E^+)$ with $v(T)=0$. 

    Note that the linear operators $\tilde{E}$ and $\tilde{S}$ are extensions in the sense that $\tilde{E}u = Eu$ and $\tilde{S}u=u(0)$ for all $u \in D(E)$. This is an immediate consequence of the adjointness formula (\ref{eq:adjointness}) and the definitions for $\tilde{E}$ and $\tilde{S}$. 
\end{proof}

\subsection{Existence of weak solutions}\label{sec:existence}

In order to address the existence of weak solutions, we derive a so-called energy estimate, cf. \cite[Lemma 5.10]{givon2005}. To this end, we assume that the numerical range of $A^\dagger$ is contained in a left half-plane. 

\begin{definition}[Numerical abscissa]
    Let $\omega$ be defined by
    \begin{align*}
        \omega &:= \sup \left\{ \frac{\Re\left((x,A^\dagger x)\right)}{(x,x)} \Big\rvert 0\neq x\in D(A^\dagger) \right\} \, . 
    \end{align*}
\end{definition}

\begin{remark*}
    The numerical range is closely linked to the growth bounds for linear evolution equations: Suppose $A^\dagger$ generates a strongly continuous semigroup $T(t)$. Then \cite[Lemma 2.2]{davies}
    \begin{equation*}
        \omega= \min\{\lambda: \|T(t)\| \leq e^{\lambda t} \text{ for all } t\geq 0\} \, .
    \end{equation*}
    In this context, the number $\omega$ is sometimes called (upper) logarithmic norm.
\end{remark*}

\begin{lemma}[Energy estimate]
    Let $\omega < \infty$. There exists a number $\gamma$ such that for all $v\in D(E^+)$
    \begin{align}
        \|v\|^2_{L^2(0,T;H)} + \|v(0)\|^2 &\leq \gamma^2 \left(\|v(T)\|^2 + \|E^+ v\|^2_{L^2(0,T;H)}\right) \, . \label{eq:energy_estimate}
    \end{align}
\end{lemma}
\begin{proof}
    \begin{align*}
        \frac{\dd }{\dd t} \|v(T-t)\|^2 &= -2 \Re \left((v'(T-t),v(T-t))\right) \\
        &= 2 \Re \left((E^+ v(T-t),v(T-t))\right) +2 \Re \left((A^\dagger v(T-t),v(T-t))\right) \\
        &\leq 2\|E^+ v(T-t)\|\|v(T-t)\|+2\omega \|v(T-t)\|^2 \\
        &\leq \|E^+ v(T-t)\|^2+\|v(T-t)\|^2+2\omega \|v(T-t)\|^2 \\
        &= \|E^+ v(T-t)\|^2 + (1+2\omega)\|v(T-t)\|^2 \, .
    \end{align*}
    Applying Gronwall's inequality \cite[p.~664, Appendix B.2.j]{evans} yields for all $t\in [0,T]$
    \begin{align}
        \|v(T-t)\|^2 & \leq e^{(1+2\omega)t}\left(\|v(T)\|^2+\|E^+ v\|^2_{L^2(0,T;H)}\right) \, . \label{eq:gronwall}
    \end{align}
    Integrating over $t\in [0,T]$ yields
    \begin{align*}
        \|v\|^2_{L^2(0,T;H)} &\leq \frac{e^{(1+2\omega)T}-1}{1+2\omega} \left(\|v(T)\|^2+\|E^+ v\|^2_{L^2(0,T;H)}\right) \, .
    \end{align*}
    Using once more eq.~(\ref{eq:gronwall}) for $t=T$, we find
    \begin{align*}
        \|v\|^2_{L^2(0,T;H)} + \|v(0)\|^2&\leq \underbrace{\left( \frac{e^{(1+2\omega)T}-1}{1+2\omega}+e^{(1+2\omega)T}\right)}_{=:\gamma^2} \left(\|v(T)\|^2+\|E^+ v\|^2_{L^2(0,T;H)}\right) \, .
    \end{align*}
\end{proof}

The following theorem generalizes the existence proof of Givon et al. \cite[Theorem 5.11]{givon2005} in a straightforward manner. The structure of the proof is very similar, but we allow for operators $A$ that are not skew-symmetric. It suffices that the numerical abscissa $\omega$ of $A^\dagger$ is finite. 

\begin{theorem}[Existence of weak solutions]\label{theorem:existence}
    Let $\omega<\infty$, $f\in L^2(0,T;H)$ and $g\in H$. Then there exists a weak solution $u \in L^2(0,T;H)$, i.e. $\tilde{E}u=f$ and $\tilde{S}u=g$. 
\end{theorem}
\begin{proof}
    The energy estimate (\ref{eq:energy_estimate}) allows us to define the scalar product 
    \begin{align}
        (v,w)_\mathcal{H} &:= (E^+ v,E^+ w)_{L^2(0,T;H)} \label{equ:scalar_produc}
    \end{align}
    on the space of functions $v\in D(E^+)$ with $v(T)=0$. (The energy estimate is required to guarantee the positive definiteness.) The completion w.r.t. this scalar product is a Hilbert space, denoted by $\mathcal{H}$. 
    

    Next, we define the linear functional $\Phi(w)$ for all $w\in D(E^+)$ with $w(T)=0$ according to
    \begin{align*}
        \Phi(w) &:= (f,w)_{L^2(0,T;H)} +(g,w(0)) \, .
    \end{align*}
    This functional is bounded on $\mathcal{H}$:
    \begin{align*}
        |\Phi(w)| &\leq  \|f\|_{L^2(0,T;H)}\|w\|_{L^2(0,T;H)} + \|g\| \|w(0)\| \\
        & \leq \sqrt{\|f\|^2_{L^2(0,T;H)}+\|g\|^2}\sqrt{\|w\|^2_{L^2(0,T;H)}+\|w(0)\|^2} \\
        &\leq \gamma \sqrt{\|f\|^2_{L^2(0,T;H)}+\|g\|^2} \|E^+ w\|_{L^2(0,T;H)} \\
        &= \gamma \sqrt{\|f\|^2_{L^2(0,T;H)}+\|g\|^2} \|w\|_\mathcal{H} \, ,
    \end{align*}
    where we used the energy estimate (\ref{eq:energy_estimate}). Since the set of functions $w\in D(E^+)$ with $w(T)=0$ is dense in $\mathcal{H}$, the linear functional $\Phi(w)$ has a unique continuous extension on $\mathcal{H}$, which we denote again by $\Phi(w)$.
    
    By the Riesz representation theorem, there exists a unique vector $v\in\mathcal{H}$ such that
    \begin{align*}
        \Phi(w) &= (v,w)_\mathcal{H} 
    \end{align*}
    for all $w\in \mathcal{H}$. By construction of the Hilbert space $\mathcal{H}$, there exists a sequence $v_n \in D(E^+)$ with $v_n(T)=0$ such that $v_n \to v$ in $\mathcal{H}$ and $E^+ v_n$ converges to some vector $u$ in $L^2(0,T;H)$. 
    Hence, for all $w\in D(E^+)$ with $w(T)=0$, we have
    \begin{align*}
        (f,w)_{L^2(0,T;H)}+(g,w
        (0)) &= (v,w)_\mathcal{H} \\
        &= \lim_n(v_n,w)_\mathcal{H} \\
        &= \lim_n(E^+ v_n,E^+ w)_{L^2(0,T;H)} \\
        &= (u,E^+ w)_{L^2(0,T;H)} \, .
    \end{align*}
    By definition, this implies that $u$ is a weak solution, i.e. $\tilde{E}u=f$ and $\tilde{S}u=g$. 
\end{proof}

\begin{remark*}
    It is worth noting that we did not introduce the so-called strong extension of the map $w\to E^+ w$ as used in \cite{givon2005}. Let us therefore show that the construction of a weak solution as presented in the proof of theorem \ref{theorem:existence} always yields the same weak solution: Suppose there exists another sequence $\tilde{v}_n \to v$ such that $E^+ \tilde{v}_n$ converges to some vector $\tilde{u}$ in $L^2(0,T;H)$. Then, we have $(u-\tilde{u},E^+ w)_{L^2(0,T;H)}=0$ for all $w\in D(E^+)$ with $w(T)=0$. Now we choose $w=v_n-\tilde{v}_n$. Taking the limit $n\to\infty$, we obtain $\|u-\tilde{u}\|_{L^2(0,T;H)}=0$. This does not, however, imply the uniqueness of a weak solution. 
\end{remark*}

\subsection{Growth bounds and uniqueness}\label{sec:uniqueness}

In this section, we show that if the numerical range of $A$ is contained in a left half-plane, there exists at most one weak solution in $D(E)$. This follows from an a priori growth bound, which is analogous to the energy estimate (\ref{eq:energy_estimate}). We emphasise, however, that the existence of weak solutions in $D(E)$ (or the uniqueness of weak solutions in $L^2(0,T;H)$) remains an open problem, c.f. sec.~\ref{sec:outlook}. 

\begin{definition}
    Let $\omega_0$ be defined by
    \begin{align*}
        \omega_0 &:= \sup \left\{ \frac{\Re\left((x,A x)\right)}{(x,x)} \Big\rvert 0\neq x\in D(A) \right\} \, . 
    \end{align*}
\end{definition}

\begin{lemma}[A priori growth bound]
    Let $\omega_0 < \infty$, $f\in L^2(0,T;H)$ and $g\in H$. Suppose there exists a weak solution $u\in D(E)$. Then
    \begin{align}
        \|u(t)\|^2 &\leq  e^{(1+2\omega_0)t}\left(\|g\|^2+\|f\|^2_{L^2(0,T;H)}\right) \, . \label{eq:growth_bound}
    \end{align}
\end{lemma}
\begin{proof}
    Similar to the proof of the energy estimate (\ref{eq:energy_estimate}), we have for a.e. $t\in[0,T]$
    \begin{align*}
                \frac{\dd }{\dd t} \|u(t)\|^2 &= 2 \Re \left((u'(t),u(t))\right) \\
        &= 2 \Re \left((E u(t),u(t))\right) +2 \Re \left((A u(t),u(t))\right) \\
        &\leq 2\|E u(t)\|\|u(t)\|+2\omega_0 \|u(t)\|^2 \\
        &\leq \|E u(t)\|^2+\|u(t)\|^2+2\omega_0 \|u(t)\|^2 \\
        &= \|E u(t)\|^2 + (1+2\omega_0)\|u(t)\|^2 \, .
    \end{align*}
    Applying Gronwall's inequality \cite[p.~664, Appendix B.2.j]{evans} yields for all $t\in [0,T]$
    \begin{equation}
        \|u(t)\|^2 \leq e^{(1+2\omega_0)t}\left(\|u(0)\|^2+\|Eu\|^2_{L^2(0,T;H)}\right) \, .\label{eq:growth_bound_2}
    \end{equation}
    Since $Eu=f$ and $u(0)=g$, the assertion follows.
\end{proof}
\begin{corollary}\label{corollary:uniqueness}
Let $\omega_0 < \infty$, $f\in L^2(0,T;H)$ and $g\in H$. There exists at most one weak solution in $D(E)$.
\end{corollary}
\begin{proof}
    Let $u,\tilde{u}\in D(E)$ with $Eu=E\tilde{u}=f$ and $u(0)=\tilde{u}(0)=g$. Then $u-\tilde{u}$ is a weak solution in $D(E)$ for $f=0$ and $g=0$. Hence, the growth bound implies $\|u(t)-\tilde{u}(t)\|=0$ for all $t\in[0,T]$.
\end{proof}

\subsection{Outlook}\label{sec:outlook}
The results obtained are not quite enough to prove the existence of \textit{unique} weak solutions. Analogous to the work by Friedrichs \cite{friedrichs1954}, it might be possible to close the gap between the existence and uniqueness under reasonable assumptions. In order to illustrate the problem, let us derive a sufficient criterion for the uniqueness of weak solutions. 

Let $E_0$ and $ E^+_0$ denote the restrictions of $E$ and $E^+$ to the set of functions with $u(0)=0$ and $v(T)=0$, respectively, i.e.
\begin{align*}
    D(E_0) &:= \{u\in D(E) : u(0)=0\} \, , \\
    D(E^+_0) &:= \{v\in D(E^+) : v(T)=0\} \, .
\end{align*}
Since $D(E) \cap H^1_0(0,T;H)$ and $D(E^+)\cap H^1_0(0,T;H)$ are dense in $L^2(0,T;H)$, the adjointness formula (\ref{eq:adjointness}) guarantees that the operators $E_0^\ast$ and ${E^+_0}^\ast$ are densely defined, where $^\ast$ denotes the adjoint on $L^2(0,T;H)$. Hence, the operators $E_0$ and $E^+_0$ are closable with closures $\overline{E_0}={E_0^\ast}^\ast$ and $\overline{E^+_0}={{E^+_0}^\ast}^\ast$, respectively.

Now let $u_1, u_2$ be two weak solutions and define $u:=u_2-u_1$. Then
\begin{align*}
    (u,E^+_0v)_{L^2(0,T;H)} &= 0 \, ,
\end{align*}
for all $v\in D(E^+_0)$. Hence, $u\in D({E^+_0}^\ast)$ and 
\begin{align*}
    {E^+_0}^\ast u &= 0 \, .
\end{align*}
Since $E^+_0 \subseteq E^\ast_0$, we have 
\begin{align*}
    \overline{E_0} =  {E_0^\ast}^\ast \subseteq {E^+_0}^\ast \, .
\end{align*}
Mind that ${E^+_0}^\ast$ might be a proper extension of $\overline{E_0}$. Now suppose there exists a sequence $u_n$ and some $\omega$ such that
\begin{align}
    D(E_0)\ni u_n \to u \in L^2(0,T;H) \, , \label{eq:sufficient_1} \\
    E_0u_n \to w \in L^2(0,T;H) \, .\label{eq:sufficient_2}
\end{align}
Then, we have $u\in D(\overline{E_0})$ and $\overline{E_0}u = {E^+_0}^\ast u=0$. Integrating the growth bound (\ref{eq:growth_bound_2}), inserting $u_n$, and taking the limit $n\to \infty$ yields 
\begin{align*}
   \|u\|^2 &\leq \frac{e^{(1+2\omega_0)T}-1}{1+2\omega_0} \left\|\overline{E_0}u \right\|^2 = 0 \, .
\end{align*}

We conclude the following statement: Let $\omega<\infty$ and $\omega_0<\infty$. Suppose for all weak solutions $u$ with $\tilde{E}u=0$ and $\tilde{S}u=0$ there exists a sequence $u_n$ satisfying eqs.~(\ref{eq:sufficient_1})(\ref{eq:sufficient_2}). Then, for all $f\in L^2(0,T;H)$ and $g\in H$, there exists a \textit{unique} weak solution in $L^2(0,T;H)$ with $\tilde{E}u=f$ and $\tilde{S}u=g$.

\section{Application to the projection operator formalism}\label{sec:application}

In this section, we show that the existence and uniqueness results from secs.~\ref{sec:existence}-\ref{sec:uniqueness} are applicable to the projection operator formalism. First, we outline the general assumptions required. In sec.~\ref{sec:zwanzig}, we  show that Zwanzig's projection is compatible with these assumptions. We demonstrate this for a nonstationary non-Hamiltonian system in sec.~\ref{sec:example}.

\begin{assumption}\label{assumption2}
    Let $U(t)$ be a strongly continuous quasicontraction semigroup on a complex Hilbert space $H$ with generator $L$. Let $P$ and $Q=1-P$ be orthogonal projections on $H$. We assume that $LQ$ as well as $L^\dagger Q$ are densely defined.
\end{assumption}
For the operator $A$, we choose 
\begin{align*}
    A:=\overline{QL}Q \, .
\end{align*}

Before we proceed, we show that assumption \ref{assumption2} implies that assumtion \ref{assumption1} is fulfilled, i.e. we have to show that $A$ is densely defined and closable. Since $L$ is a generator, it is densely defined. Hence, $QL$ is densely defined. The adjoint is given by $[QL]^\dagger=L^\dagger Q$ \cite[p.~330, Theorem 13.2]{rudin1974}. Since $L^\dagger Q$ is densely defined, the adjoint of $QL$ is densely defined. Hence, $QL$ is closable. It is easy to show that the operator $A=\overline{QL}Q$ is again a closed operator. Since $LQ$ is densely defined, $\overline{QL}Q$ is also densely defined. 

Moreover, assumption \ref{assumption2} guarantees that there exists some $\lambda\in\mathbb{R}$ such that
\begin{align}
    \|U(t)x\| &\leq e^{\lambda t}\|x\| \label{eq:quasi_contraction}
\end{align}
for all $x \in H$ and $t\geq 0$. According to Davies \cite[Lemma 2.2]{davies}, the smallest constant $\lambda$ for which eq.~(\ref{eq:quasi_contraction}) holds is given by 
\begin{align*}
    \lambda_{\min} := \sup \{\Re(z)|z\in \num(L)\} <\infty\, ,
\end{align*}
where $\num(L)$ is the numerical range of $L$. 

\begin{lemma}[Existence of weak solutions]\label{lemma:existence_2}
    There exists a weak solution $u\in L^2(0,T;H)$ for all $f\in L^2(0,T;H)$ and $g\in H$. 
\end{lemma}
\begin{proof}
    We have to show that
    \begin{align*}
        \omega := \sup \left\{\Re(z)| z \in \num(A^\dagger)\right\} &< \infty \, .
    \end{align*}
    The adjoint of $A$ is given by
    \begin{align*}
       A^\dagger = (\overline{QL}Q)^\dagger &= ({[QL]^\dagger}^\dagger Q)^\dagger = {( Q[QL]^\dagger)^\dagger}^\dagger = \overline{QL^\dagger Q} \, ,
    \end{align*}
    where we used \cite[p.~252, Theorem VIII.1]{reed1980} and \cite[p.~330, Theorem 13.2]{rudin1974}. We have
    \begin{align*}
        \num(A^\dagger) &\subseteq  \overline{ \num(QL^\dagger Q)}  \subseteq \overline{ \num(L^\dagger )} \, ,
    \end{align*}
    where we used the fact that $\num(\overline{T})\subseteq \overline{\num(T)}$ for any closable linear operator $T$. Hence, 
    \begin{align*}
        \omega &\leq \sup \left\{\Re(z)| z \in \num(L^\dagger )\right\} \, .
    \end{align*}
    
    Since $H$ is a Hilbert space and therefore reflexive, the adjoint semigroup $U(t)^\dagger$ is again a strongly continuous quasicontraction semigroup with generator $L^\dagger$ \cite[p.~5-6]{vanNeerven1992}. Furthermore, it satisfies the same growth bounds, since $\|U(t)^\dagger\|= \|U(t)\|$. Hence, according to Davies \cite[Lemma 2.2]{davies}, we have
    \begin{align*}
        \lambda_{\min} &= \sup \left\{\Re(z) | z\in \num(L) \right\} = \sup \left\{\Re(z) | z\in \num(L^\dagger) \right\} <\infty \, .
    \end{align*}
    Therefore, we find
    \begin{align*}
        \omega & < \infty \, .
    \end{align*}

    Hence, the assumptions of theorem \ref{theorem:existence} are fulfilled, and for all $f\in L^2(0,T;H)$ and $g\in H$ there exists a weak solution in $L^2(0,T;H)$.  
\end{proof}

\begin{lemma}[Uniqueness of weak solutions]\label{lemma:uniqueness_2}
    For all $f\in L^2(0,T;H)$ and $g\in H$, there exists at most one weak solution $u$ such that $u\in H^1(0,T;H)$ and $\overline{QLQ} u \in L^2(0,T;H)$.  
\end{lemma}
\begin{proof}
    We would like to show $\omega_0<\infty$ in order to apply the uniqueness result from corollary \ref{corollary:uniqueness}. However, it appears to be difficult to establish the desired estimate for the numerical abscissa $\omega_0$ of $A=\overline{QL}Q$ under the given assumptions. Hence, we only show that the numerical abscissa $\tilde{\omega}_0$ for the operator $\tilde{A}:=\overline{QLQ}$ is finite. We have already shown that $QL$ is closable. This implies that $QLQ$ is also closable. The closures are related by $\tilde{A}=\overline{QLQ}\subseteq \overline{QL}Q=A$. 

    We have to show that
    \begin{align*}
       \tilde{\omega}_0 &:=\sup\{\Re(z)|z\in\num(\tilde{A})\} < \infty \, .
    \end{align*}
    For the numerical range of $\tilde{A}$, we find
    \begin{align*}
        \num(\overline{QLQ}) \subseteq \overline{\num(QLQ)} \subseteq \overline{\num(L)} \, .
    \end{align*}
    This implies 
    \begin{align*}
        \tilde{\omega}_0 = \sup\{\Re(z)|z\in\num(QLQ)\} \leq \sup\{\Re(z)|z\in\num (L)\} =  \lambda_{\min}<\infty\, .
    \end{align*}
    
    Since $\tilde{A}\subseteq A$, any weak solution in $H^1(0,T;H)$ such that $\tilde{A}u\in L^2(0,T;H)$ satisfies the growth bound (\ref{eq:growth_bound}).     Hence, the same argument as in the proof of corollary (\ref{corollary:uniqueness}) shows that there exists at most one weak solution $u$ such that $u\in H^1(0,T;H)$ and $\overline{QLQ} u \in L^2(0,T;H)$.  
\end{proof}

We emphasize that there are examples for which $\overline{QL}Q$ is a proper extension of $\overline{QLQ}$. Such an example is given by Givon et al. \cite[Section 3]{givon2005}, which is also discussed in \cite[Example 3.10]{widder2025}. 

\subsection{Zwanzig's projection operator}\label{sec:zwanzig}

Let $\rho \in L^1(\mathbb{R}^n,\mathcal{B}(\mathbb{R}^n),\mu)$ be nonnegative and normalized, where $\mathcal{B}(\mathbb{R}^n)$ is the Borel $\sigma$-algebra and $\mu$ is the Lebesgue measure. For the Hilbert space $H$, we choose 
\begin{align*}
    H&:= L^2_\rho(\mathbb{R}^n) := L^2(\mathbb{R}^n,\mathcal{B}(\mathbb{R}^n), \rho\mu) \, ,
\end{align*}
where the probability measure $\rho\mu:\mathcal{B}(\mathbb{R}^n)\to[0,1]$ is defined according to
\begin{align*}
    \rho\mu(V) &:= \int_V \rho \,\dd\mu \, .
\end{align*}
The scalar product on $L^2_\rho(\mathbb{R}^n)$ is given by
\begin{align*}
    (u,v) &:= \int u^\ast v \, \rho \, \dd \mu \, .
\end{align*}
In the following, we write $\rho=\rho(\mathbf{x},\mathbf{y})$ for $\mathbf{x}\in \mathbb{R}^m$ and $\mathbf{y}\in\mathbb{R}^{n-m}$. 

We can now define the Zwanzig projection operator according to 
\begin{align}
    Pu &:= \frac{\int  u(\cdot,\mathbf{y}) \rho(\cdot,\mathbf{y}) \,\dd \mathbf{y}  }{\int \rho(\cdot,\mathbf{y}) \,\dd \mathbf{y}} \, . \label{equ:zwanzig}
\end{align}
Note that $Pu$ may be undefined on sets with probability measure zero. The self-adjointness of $P$ is easily established using Fubini's theorem. Clearly, $P$ is idempotent. Hence, $P$ is an orthogonal projection on $L^2_\rho(\mathbb{R}^n)$. 

We are interested in cases where $L$ and $L^\dagger$ are first-order differential operators. In these cases, the following lemma establishes sufficient conditions such that the operators $LQ$ and $L^\dagger Q$ are densely defined. Let $R(Q)$ denote the range of $Q$.

\begin{lemma}\label{lemma:dense_subspace}
    Let $\rho \in C^1(\mathbb{R}^n)$ be nonnegative and normalized. Then  $C^1_c(\mathbb{R}^n) \cap R(Q)$ is dense in $R(Q)$.
\end{lemma}
\begin{proof}
    The proof relies on a slight modification of the construction used in \cite[Lemma 5.2]{givon2005}.

    We show that a function $u\in R(Q)$ can be approximated by a function $w \in C^1_c(\mathbb{R}^n) \cap R(Q)$. Let $u\in R(Q)$ and $\varepsilon>0$. Since $C^1_c(\mathbb{R}^n)$ is dense in $L^2_\rho(\mathbb{R}^n)$, there exists a function $v \in C^1_c(\mathbb{R}^n)$ such that $\|u-v\| \leq \varepsilon$.
    Next, we truncate and mollify $v$ to obtain a function $\tilde{v}\in C^1_c(\mathbb{R}^n)$ such that the support of $\tilde{v}$ is contained in the interior of the support of $\rho$ and $\|v-\tilde{v}\|<\varepsilon$. This is possible, because the mollification converges uniformly on compact sets \cite[p.~108, Proposition 4.21]{Brezis2011}. We have
    \begin{align*}
        \|u-\tilde{v}\| &\leq 2\varepsilon \, ,\\
        \|P\tilde{v}\| &\leq 2\varepsilon \, , \\
        \text{supp}(\tilde{v})&\subseteq \text{int}(\text{supp}(\rho)) \, .
    \end{align*}
    Now let $\eta(\mathbf{x},\cdot)$ be the mollification of a $(n-m)$-dimensional cube, i.e.
    \begin{align*}
        \eta(\mathbf{x},\cdot) &:= \tilde{\eta} \star \chi_{[-R(\mathbf{x}), R(\mathbf{x})]^{n-m}} \, ,
      \end{align*}
    where $\tilde{\eta}$ is a mollifier on $\mathbb{R}^{n-m}$, $\chi$ is the characteristic function, $R(\mathbf{x})$ is continuously differentiable and $\star$ is the convolution. For $\tilde{\eta}$ we choose the normalized bump function on $\mathbb{R}
    ^{n-m}$, such that $\eta \leq 1$ on $\mathbb{R}^n
    $. Clearly, $R(\mathbf{x})$ can be chosen such that
    \begin{align*}
        \int_{\mathbb{R}^{n-m}} \eta(\mathbf{x},\mathbf{y}) \rho(\mathbf{x},\mathbf{y}) \dd \mathbf{y} &\geq \frac{1}{2}\int_{\mathbb{R}^{n-m}} \rho(\mathbf{x},\mathbf{y}) \dd \mathbf{y} 
    \end{align*}
    for all $\mathbf{x}\in\mathbb{R}^m$. This implies that $\frac{\eta(\mathbf{x},\mathbf{y})}{(P\eta)(\mathbf{x})} \leq 2$ for all $\mathbf{x},\mathbf{y}$ with $\int \rho(\mathbf{x},\tilde{\mathbf{y}}) \dd \tilde{\mathbf{y}} >0$. 
    
    Finally, let us define
    \begin{align*}
        w &:=
            \tilde{v} - \eta\frac{P\tilde{v}}{P\eta} \,  .
    \end{align*}
    We have to show that $w \in C^1_c(\mathbb{R}^n)\cap R(Q)$. Clearly, $\eta$ is continuously differentiable. Using the dominated convergence theorem, it is straightforward to show that the integrals
    \begin{align*}
        \int_{\mathbb{R}^{n-m}} \eta(\mathbf{x},\mathbf{y}) \rho(\mathbf{x},\mathbf{y}) \dd \mathbf{y} \, , \\
        \int_{\mathbb{R}^{n-m}} \tilde{v}(\mathbf{x},\mathbf{y}) \rho(\mathbf{x},\mathbf{y}) \dd \mathbf{y}
    \end{align*}
    are continuously differentiable, cf. the Leibniz integral rule \cite[p.~56, Theorem 2.27]{folland}. This shows that $\frac{P\tilde{v}}{P\eta}$ is continuously differentiable, since the support of $P\tilde{v}$ is by construction contained in the interior of the support of $P\eta$. Hence, $w$ is continuously differentiable. We note that $P\tilde{v}$ has compact support as a function of $\mathbf{x}\in \mathbb{R}^m$ and $\eta(\mathbf{x},\cdot)$ has compact support as a function of $\mathbf{y}\in \mathbb{R}^{n-m}$ for all $\mathbf{x}$. Moreover, there exists a compact set in $\mathbb{R}^{n-m}$ that contains the support of $\eta(\mathbf{x},\cdot)$ for any $\mathbf{x}$ in the support of $P\tilde{v}$. This shows that $\eta P\tilde{v}$ has compact support on $\mathbb{R}^n$. Hence, $w\in C^1_c(\mathbb{R}^n)$. Clearly, we have $Pw=0$, i.e. $w\in R(Q)$. Thus, we have $w \in C^1_c(\mathbb{R}^n)\cap R(Q)$ and 
    \begin{align*}
        \|u-w\| &\leq \|u-\tilde{v}\| + \|\tilde{v}-w\| \\
        &\leq \|u-\tilde{v}\| + 2\left\|P\tilde{v}\right\| \\
        &\leq 6\varepsilon \, .
    \end{align*}
    This concludes the proof. 
\end{proof}

\subsection{Example}\label{sec:example}

In order to apply our results, we require a dynamical system that gives rise to a quasicontraction semigroup $U(t)$ for the time-evolution of random variables in $H=L^2_\rho(\mathbb{R}^n)$. To this end, we use the following lemma.

\begin{lemma}\label{lemma:time_evolution}
    Let $\rho \in L^1(\mathbb{R}^n)$ be nonnegative and normalized. Let $\mathbf{F}\colon\mathbb{R}^n\to\mathbb{R}^n$ be continuously differentiable with bounded derivative. We further assume $\rho\mathbf{F}\in W^{1,1}(\mathbb{R}^n)$ and 
    \begin{align*}
        \lambda &:= \sup_{ \{\Gamma\in\mathbb{R}^n | \rho(\Gamma)\neq 0 \}} \left(-\frac{1}{2\rho}{\rm div}(\rho\mathbf{F}) \right) < \infty.
    \end{align*}
    Then the closure of $(\mathbf{F}\cdot\boldsymbol{\nabla},C^1_c(\mathbb{R}^n))$ in $L^2_\rho(\mathbb{R}^n)$ generates a strongly continuous semigroup $U(t)$ such that $\|U(t)\| \leq e^{\lambda t}$ and $U(t)x=x(\varphi_t)$ for all $x\in C^1_c(\mathbb{R}^n)$ and $t\geq 0$, where $\varphi_t$ denotes the flow generated by the vector field $\mathbf{F}$.
\end{lemma}
\begin{proof}
    The proof is given in appendix \ref{app:proof}.
\end{proof}

Let us consider a simple nonstationary non-Hamiltonian system: the one-dimensional damped harmonic oscillator. The vector field $\mathbf{F}$ takes the form
\begin{align*}
    \mathbf{F}(q,p) &= 
    \begin{pmatrix}
        p \\
        -q-2\gamma p  
    \end{pmatrix} \, ,
\end{align*}
where $\gamma>0$ is the damping ratio. The easiest phase space density would be a Gaussian distribution. To illustrate the scope of the results obtained, we choose the 'bump function' instead: 
\begin{align*}
    \rho(q,p) &:= \begin{cases}
        C \exp\left( \frac{1}{q^2+p^2-1}\right) &\quad q^2+p^2 < 1  \\
        0 &\quad \text{else} \, ,
    \end{cases}
\end{align*}
where $C$ is a normalization constant. In order to obtain a quasicontraction semigroup, we show that $\lambda < \infty$:
\begin{align*}
    \lambda &:= \sup_{ \{\Gamma\in\mathbb{R}^n | \rho(\Gamma)\neq 0 \}} \left(-\frac{1}{2\rho}{\rm div}(\rho\mathbf{F}) \right) 
    = \sup_{q^2+p^2<1} \left(\gamma - \frac{2\gamma p^2}{(q^2+p^2-1)^2}\right) 
    =\gamma \, .
\end{align*}
Hence, the conditions of lemma \ref{lemma:time_evolution} are fulfilled and
\begin{align*}
    L &:= \overline{(\mathbf{F}\cdot \boldsymbol{\nabla},C^1_c(\mathbb{R}^n))}
\end{align*}
generates a strongly continuous semigroup $U(t)$ with
\begin{align*}
    \|U(t)\| &\leq e^{\gamma t} 
\end{align*}
for all $t\geq 0$. 
For the Zwanzig projection operator, we choose
\begin{align*}
    Pu &:= \frac{\int  u(\cdot,p) \rho(\cdot,p) \,\dd p  }{\int \rho(\cdot,p) \,\dd p} \, .
\end{align*}

Lemma \ref{lemma:dense_subspace} now guarantees that $LQ$ as well as $L^\dagger Q$ are densely defined: Since $C^1_c(\mathbb{R}^n)\cap R(Q)$ is dense in $R(Q)$, $LQ$ is densely defined. Now let $u\in D(L)$ and $v \in C^1_c(\mathbb{R}^n)\cap R(Q)$. Since $L$ is the closure of $(\mathbf{F}\cdot\boldsymbol{\nabla},C^1_c(\mathbb{R}^n))$, there exists a sequence $C^1_c(\mathbb{R}^n) \ni u_n \to u$ such that $\mathbf{F}\cdot\boldsymbol{\nabla} u_n \to Lu$. Hence,
\begin{align*}
    (Lu,v) &= \lim_n (\mathbf{F}\cdot \boldsymbol{\nabla}u_n,v) = - \lim_n \left(u_n, \frac{1}{\rho} \text{div}(\rho\mathbf{F}v)\right) = - \left(u, \frac{1}{\rho} \text{div}(\rho\mathbf{F}v)\right) \, .
\end{align*}
This shows that $v\in D(L^\dagger)$ with $L^\dagger v = -\frac{1}{\rho}\text{div}(\rho\mathbf{F}v)$ for all $v\in C^1_c(\mathbb{R}^n)\cap R(Q)$. Since $C^1_c(\mathbb{R}^n)\cap R(Q)$ is dense in $R(Q)$, this shows that $L^\dagger Q$ is densely defined. 

Therefore, the assumptions of sec.~\ref{sec:application} are fulfilled, and we can immediately apply lemmata \ref{lemma:existence_2} and \ref{lemma:uniqueness_2}. 

\section{Discussion and conclusion}\label{sec:conclusion}
Let $H$ be a complex Hilbert space. We proved the existence of weak solutions $u\in L^2(0,T;H)$ in the spirit of Friedrichs' extension for the abstract Cauchy problem (\ref{eq:cauchy_problem}) associated to some densely defined closable operator $A$ whose adjoint admits a finite  numerical abscissa (Theorem \ref{theorem:existence}). The a priori growth bounds and the uniqueness of these solutions, on the other hand, appear to be more subtle. If the operator $A$ admits a finite  numerical abscissa, the growth bound as well as the uniqueness of a weak solution $u$ follows, provided that  $u$ is sufficiently regular, i.e. $u\in H^1(0,T;H)$ and $Au \in L^2(0,T;H)$ (sec.\ref{sec:uniqueness}). The existence of sufficiently regular solutions, however, remains unclear. Hence, the existence of \textit{unique} weak solutions remains unsolved (sec.~\ref{sec:outlook}). Furthermore, it can be challenging to obtain estimates on the numerical abscissas for $A$ and $A^\dagger$ simultaneously. This becomes apparent when applying the results obtained to the projection operator formalism (sec.~\ref{sec:application}). 

Let $U(t)$ be a strongly continuous quasicontraction semigroup with generator $L$. Let $Q$ be an orthogonal projection and let $A:=\overline{QL}Q$. Then there exists a weak solution to the ACP (\ref{eq:cauchy_problem}), if $LQ$ as well as $L^\dagger Q$ are densely defined (Lemma \ref{lemma:existence_2}). In essence, this result follows from the fact that $A^\dagger$ has a finite  numerical abscissa under the given assumptions. 
We note that it remains unclear whether the weak solution $u(t)$ for $f=0$ and $g\in H$ is regular enough to yield a variation of constants formula analogous to eq.~(\ref{equ:variation_of_constants}). For most infinite-rank projections $P$, the operator $PL$ is unbounded and does not possess a continuous extension on $H$. Hence, for the variation of constants formula to make sense, a necessary condition is $u(t)\in D(\overline{PL})$ for almost every $t\in[0,T]$. It remains an open question which initial values $g\in H$ give rise to this property.

For the uniqueness of weak solutions, it can make a difference if we consider the operator $\overline{QL}Q$ or $\overline{QLQ}$, since there are cases where $\overline{QL}Q$ is a proper extension of $\overline{QLQ}$. Consequently, we are not able to show that the numerical abscissa of $A:=\overline{QL}Q$ is finite. Instead, we showed that the numerical abscissa of $\overline{QLQ}$ is finite. With this, the growth bounds as well as the uniqueness of a weak solution $u$ follow, if $u\in H^1(0,T;H)$ and $\overline{QLQ} u \in L^2(0,T;H)$ (Lemma \ref{lemma:uniqueness_2}). Hence, we had to exclude functions in the domain of $\overline{QL}Q$ that are not contained in the domain of $\overline{QLQ}$.

Finally, we considered the Zwanzig projection operator $P$ on the space of square integrable random variables $H=L^2_\rho(\mathbb{R}^n)$ defined by eq.~(\ref{equ:zwanzig}). In order to apply the results obtained, it is necessary to show that $LQ$ as well as $L^\dagger Q$ are densely defined. To this end, we showed that there exists a dense set of compactly supported and continuously differentiable functions in the range of $Q$, provided the initial phase space density $\rho$ is continuously differentiable (Lemma \ref{lemma:dense_subspace}). As an example, we applied our results to the damped harmonic oscillator using the bump function as initial distribution (sec.~\ref{sec:example}).

The main contribution of this article is the generalization of the existence proof for weak solutions of the orthogonal dynamics equation by Givon et al. \cite{givon2005}. The main ideas of the proofs remain analogous, but the requirement that $QLQ$ is skew-symmetric is relaxed. This enables us to proof the existence of weak solutions to the orthogonal dynamics equation for a large class of nonstationary non-Hamiltonian systems and any orthogonal projection $Q$ with the property that the range of $Q$ contains a dense set of compactly supported and continuously differentiable functions. 

\begin{acknowledgments}
     The author would like to thank Tanja Schilling and Fabian Koch for helpful discussions. The author acknowledges funding by the German Research Foundation (DFG) in Project 535083866. 
\end{acknowledgments}

\appendix

\section{Proof of lemma \ref{lemma:adjointness}}\label{app:integration_by_parts}
\begin{proof}
Let us briefly address integration by parts. Let $u,v \in H^1(0,T;H)$, where $H$ is a Hilbert space. Then, the functions $u$ and $v$ are absolutely continuous. Hence, the map $t\to (u(t),v(t))$ is absolutely continuous. Thus, $t\to (u(t),v(t))$ is weakly differentiable and its weak derivative coincides with the pointwise derivative almost everywhere. The pointwise derivative can be evaluated using the product rule almost everywhere, because $u$ and $v$ are weakly differentiable and therefore strongly differentiable almost everywhere. Hence, by the fundamental theorem of Lebesgue integral calculus, we have
\begin{align*}
     (u(T),v(T)) &= (u(0),v(0))+\int^T_0 (u'(t),v(t))+(u(t),v'(t)) \, \dd t \, .\
\end{align*}
With this, we obtain the adjointness formula (\ref{eq:adjointness}),
    \begin{align*}
    (Eu,v)_{L^2(0,T;H)} &= \int^T_0 (Eu(t),v(t)) \, \dd t \\
    &= \int^T_0 (u'(t),v(t))-(Au(t),v(t)) \,\dd t \\
    &= (u(T),v(T))-(u(0),v(0))- \int^T_0 (u(t),v'(t))+(u(t),A^\dagger v(t))\, \dd t \\
    &= (u(T),v(T))-(u(0),v(0)) + (u,E^+ v)_{L^2(0,T;H)} \, ,
    \end{align*}
    for all $u\in D(E)$ and $v\in D(E^+)$. 
\end{proof}

\section{Proof of lemma \ref{lemma:dense_subspace}}\label{app:dense_subspace}
\begin{proof}
    We only proof the first statement. The proof for the second statement is identical.
    
    First, we show that the set of simple functions is dense in the Bochner space $L^2(0,T;H)$. A function $s$ is called simple, if
    \begin{align*}
        s &= \sum^n_{i=1} \chi_{E_i}x_i \, ,
    \end{align*}
    where $x_i \in H$, $\chi$ is the characteristic function and the sets $E_i \subseteq [0,T]$ are measurable and pairwise disjoint. For any strongly measurable function $f:[0,T]\to H$, there exists a sequence of simple functions $s_k$ such that $s_k(t) \to f(t)$ and $\|s_k(t)\|\leq \|f(t)\|$ for a.e. $t\in [0,T]$ \cite[p.~398, E.2]{cohn}. Hence, for a.e. $t\in [0,T]$, we have
    \begin{align*}
        \|s_k(t)-f(t)\|^2 &\to 0  \, , \\
        \|s_k(t)-f(t)\|^2 &\leq (\|s_k(t)\|+\|f(t)\|)^2 \leq 4 \|f(t)\|^2 \, .
    \end{align*}
    By the dominated convergence theorem, we have
    \begin{align*}
       \lim_{k\to\infty}\int^T_0 \|s_k(t)-f(t)\|^2 \, \dd t &=   \int^T_0\lim_{k\to\infty} \|s_k(t)-f(t)\|^2 \, \dd t = 0 \, .
    \end{align*}
    This shows that the space of simple functions is dense in $L^2(0,T;H)$. 
    
    Let $s$ be a simple function. Since $D(A)$ is dense in $H$, for each $x_i$ there exists a sequence $D(A)\ni x_{ik}\to x_i \in H$. Hence, the sequence
    \begin{align*}
        s_k &:= \sum^n_{i=1} \chi_{E_i}x_{ik} 
    \end{align*}
    converges to $s$ in $L^2(0,T;H)$. This shows that the set of simple functions with values in $D(A)$ is dense in the set of simple functions. Hence, the set of simple functions with values in $D(A)$ is dense in $L^2(0,T;H)$. 

    The last step of the proof is accomplished by a simple mollification. Mollifications of vector-valued functions are also used, for instance, in the proof of the fundamental theorem of Bochner integral calculus \cite[p.~285-286, Theorem 2]{evans}. Let $s$ be a simple function with values in $D(A)$. We truncate and mollify this function:
    \begin{align*}
        s_\varepsilon := \eta_\varepsilon \star (s\cdot \chi_{[\varepsilon,T-\varepsilon]}) \, ,
    \end{align*}
    where $\eta_\varepsilon$ is a mollifier with support $[-\varepsilon,\varepsilon]$ and $\star$ is the convolution. Then $s_\varepsilon\in D(E) \cap H^1_0(0,T;H)$ and $s_\varepsilon \to s$ in $L^2(0,T;H)$. To see this, recall that $s$ has only a finite number of values. Hence, the convergence is easily inferred from the fact that the mollification of a real-valued function in $L^2(\mathbb{R})$ converges in $L^2(\mathbb{R})$ \cite[p.~109, Theorem 4.22]{Brezis2011}. Since the set of simple functions with values in $D(A)$ is dense in $L^2(0,T;H)$, this shows that $D(E) \cap H^1_0(0,T;H)$ is dense in $L^2(0,T;H)$. 
\end{proof}

\section{Proof of Lemma \ref{lemma:time_evolution}}\label{app:proof}
\begin{proof}
    For all $x\in C^1_c(\mathbb{R}^n)$, integration by parts yields \cite[Eq.~B.2]{widder2025}
    \begin{align*}
        \frac{\dd}{\dd t}\|x(\varphi_t)\|^2 &= - \left(\frac{1}{\rho}\text{div}(\rho\mathbf{F})x(\varphi_t),x(\varphi_t)\right) \leq 2\lambda \|x(\varphi_t)\|^2 \, ,
    \end{align*}
    where we used the monotonicity of the integral in the second step. Due to Gronwall's inequality, this implies 
    \begin{align*}
        \|x(\varphi_t)\| &\leq e^{\lambda t} \|x\| \, ,
    \end{align*}
    for all $t\geq0$. 
    
    Since $C^1_c(\mathbb{R}^n)$ is dense in $L^2_\rho(\mathbb{R}^n)$, the map $x\to x(\varphi_t)$ extends to a strongly continuous semigroup $U(t)$ on $L^2_\rho(\mathbb{R}^n)$ and the generator is identified with the closure of $(\mathbf{F}\cdot\boldsymbol{\nabla},C^1_c(\mathbb{R}^n))$ in $L^2_\rho(\mathbb{R}^n)$. To show this, it suffices to repeat the line of argument in the proof of \cite[Theorem~4.1]{widder2025}.
\end{proof}


\bibliographystyle{unsrt}
\bibliography{refs}

@article{li2015,
    author = {Li, Zhen and Bian, Xin and Li, Xiantao and Karniadakis, George Em},
    title = {Incorporation of memory effects in coarse-grained modeling via the {Mori}-{Zwanzig} formalism},
    journal = {The Journal of Chemical Physics},
    volume = {143},
    number = {24},
    pages = {243128},
    year = {2015},
    month = {11},
    issn = {0021-9606},
    doi = {10.1063/1.4935490},
    url = {https://doi.org/10.1063/1.4935490},
    eprint = {https://pubs.aip.org/aip/jcp/article-pdf/doi/10.1063/1.4935490/15508956/243128_1_online.pdf},
}

@article{li2010,
author = {Li, Xiantao},
title = {A coarse-grained molecular dynamics model for crystalline solids},
journal = {International Journal for Numerical Methods in Engineering},
volume = {83},
number = {8-9},
pages = {986-997},
doi = {https://doi.org/10.1002/nme.2892},
url = {https://onlinelibrary.wiley.com/doi/abs/10.1002/nme.2892},
eprint = {https://onlinelibrary.wiley.com/doi/pdf/10.1002/nme.2892},
year = {2010}
}

@book{hansen2013,
  title={Theory of Simple Liquids: with Applications to Soft Matter},
  author={Hansen, J.P. and McDonald, I.R.},
  isbn={9780123870339},
  lccn={2013372077},
  year={2013},
  publisher={Academic Press}
}

@article{ahmadirahmat2025,
  title = {Mode-coupling theory of the glass transition for a liquid in a periodic potential},
  author = {Ahmadirahmat, Abolfazl and Caraglio, Michele and Krakoviack, Vincent and Franosch, Thomas},
  journal = {Phys. Rev. E},
  volume = {112},
  issue = {1},
  pages = {015405},
  numpages = {15},
  year = {2025},
  month = {Jul},
  publisher = {American Physical Society},
  doi = {10.1103/ks5t-xtvd},
  url = {https://link.aps.org/doi/10.1103/ks5t-xtvd}
}

@article{vrugt2019,
  title = {{Mori}-{Zwanzig} projection operator formalism for far-from-equilibrium systems with time-dependent {Hamiltonians}},
  author = {te Vrugt, Michael and Wittkowski, Raphael},
  journal = {Phys. Rev. E},
  volume = {99},
  issue = {6},
  pages = {062118},
  numpages = {26},
  year = {2019},
  month = {Jun},
  publisher = {American Physical Society},
  doi = {10.1103/PhysRevE.99.062118},
  url = {https://link.aps.org/doi/10.1103/PhysRevE.99.062118}
}

@article{nakajima1958,
    author = {Nakajima, Sadao},
    title = {On Quantum Theory of Transport Phenomena: Steady Diffusion},
    journal = {Progress of Theoretical Physics},
    volume = {20},
    number = {6},
    pages = {948-959},
    year = {1958},
    month = {12},
    issn = {0033-068X},
    doi = {10.1143/PTP.20.948},
    url = {https://doi.org/10.1143/PTP.20.948},
    eprint = {https://academic.oup.com/ptp/article-pdf/20/6/948/5440766/20-6-948.pdf},
}

@article{koch2024,
  title = {Nonequilibrium solvent response force: What happens if you push a {Brownian} particle},
  author = {Koch, Fabian and Erle, Jona and Schilling, Tanja},
  journal = {Phys. Rev. Res.},
  volume = {6},
  issue = {1},
  pages = {L012032},
  numpages = {5},
  year = {2024},
  month = {Feb},
  publisher = {American Physical Society},
  doi = {10.1103/PhysRevResearch.6.L012032},
  url = {https://link.aps.org/doi/10.1103/PhysRevResearch.6.L012032}
}

@article{meyer2017,
    author = {Meyer, Hugues and Voigtmann, Thomas and Schilling, Tanja},
    title = {On the non-stationary generalized {Langevin} equation},
    journal = {The Journal of Chemical Physics},
    volume = {147},
    number = {21},
    pages = {214110},
    year = {2017},
    month = {12},
    issn = {0021-9606},
    doi = {10.1063/1.5006980},
    url = {https://doi.org/10.1063/1.5006980},
    eprint = {https://pubs.aip.org/aip/jcp/article-pdf/doi/10.1063/1.5006980/15535580/214110_1_online.pdf},
}

@article{furukawa1979,
    author = {Furukawa, Hiroshi},
    title = {Study of {Langevin} Type Equations by Means of a New Projection Operator Method in Nonequilibrium States},
    journal = {Progress of Theoretical Physics},
    volume = {62},
    number = {1},
    pages = {70-90},
    year = {1979},
    month = {07},
    issn = {0033-068X},
    doi = {10.1143/PTP.62.70},
    url = {https://doi.org/10.1143/PTP.62.70},
    eprint = {https://academic.oup.com/ptp/article-pdf/62/1/70/5464160/62-1-70.pdf},
}

@article{xing2011,
    author = {Xing, Jianhua and Kim, K. S.},
    title = {Application of the projection operator formalism to non-{Hamiltonian} dynamics},
    journal = {The Journal of Chemical Physics},
    volume = {134},
    number = {4},
    pages = {044132},
    year = {2011},
    month = {01},
    issn = {0021-9606},
    doi = {10.1063/1.3530071},
    url = {https://doi.org/10.1063/1.3530071},
    eprint = {https://pubs.aip.org/aip/jcp/article-pdf/doi/10.1063/1.3530071/13788369/044132_1_online.pdf},
}

@article{zwanzig1960,
    author = {Zwanzig, Robert},
    title = {Ensemble Method in the Theory of Irreversibility},
    journal = {The Journal of Chemical Physics},
    volume = {33},
    number = {5},
    pages = {1338-1341},
    year = {1960},
    month = {11},
    issn = {0021-9606},
    doi = {10.1063/1.1731409},
    url = {https://doi.org/10.1063/1.1731409},
    eprint = {https://pubs.aip.org/aip/jcp/article-pdf/33/5/1338/18820045/1338_1_online.pdf},
}

@article{izvekov2019,
    author = {Izvekov, Sergei},
    title = {Microscopic derivation of coarse-grained, energy-conserving generalized {Langevin} dynamics},
    journal = {The Journal of Chemical Physics},
    volume = {151},
    number = {10},
    pages = {104109},
    year = {2019},
    month = {09},
    issn = {0021-9606},
    doi = {10.1063/1.5096655},
    url = {https://doi.org/10.1063/1.5096655},
    eprint = {https://pubs.aip.org/aip/jcp/article-pdf/doi/10.1063/1.5096655/13601833/104109_1_online.pdf},
}

@article{grabert1977,
author={Grabert, H.},
title={Microdynamics and equations of motion for macrovariables},
journal={Zeitschrift f{\"u}r Physik B Condensed Matter},
year={1977},
month={Mar},
day={01},
volume={27},
number={1},
pages={95-99},
issn={1431-584X},
doi={10.1007/BF01315510},
url={https://doi.org/10.1007/BF01315510}
}

@article{darve2009,
author = {Eric Darve  and Jose Solomon  and Amirali Kia },
title = {Computing generalized {Langevin} equations and generalized {Fokker}–{Planck} equations},
journal = {Proceedings of the National Academy of Sciences},
volume = {106},
number = {27},
pages = {10884-10889},
year = {2009},
doi = {10.1073/pnas.0902633106},
URL = {https://www.pnas.org/doi/abs/10.1073/pnas.0902633106},
eprint = {https://www.pnas.org/doi/pdf/10.1073/pnas.0902633106}}

@article{ayaz2022,
  title = {Generalized {Langevin} equation with a nonlinear potential of mean force and nonlinear memory friction from a hybrid projection scheme},
  author = {Ayaz, Cihan and Scalfi, Laura and Dalton, Benjamin A. and Netz, Roland R.},
  journal = {Phys. Rev. E},
  volume = {105},
  issue = {5},
  pages = {054138},
  numpages = {19},
  year = {2022},
  month = {May},
  publisher = {American Physical Society},
  doi = {10.1103/PhysRevE.105.054138},
  url = {https://link.aps.org/doi/10.1103/PhysRevE.105.054138}
}

@article{vroylandt2022,
doi = {10.1209/0295-5075/acab7d},
url = {https://doi.org/10.1209/0295-5075/acab7d},
year = {2022},
month = {dec},
publisher = {EDP Sciences, IOP Publishing and Società Italiana di Fisica},
volume = {140},
number = {6},
pages = {62003},
author = {Vroylandt, Hadrien},
title = {On the derivation of the generalized {Langevin} equation and the fluctuation-dissipation theorem},
journal = {Europhysics Letters}
}

@article{hery2024,
doi = {10.1088/1751-8121/ad91ff},
url = {https://doi.org/10.1088/1751-8121/ad91ff},
year = {2024},
month = {nov},
publisher = {IOP Publishing},
volume = {57},
number = {50},
pages = {505003},
author = {Héry, Benjamin J A and Netz, Roland R},
title = {Derivation of a generalized {Langevin} equation from a generic time-dependent {Hamiltonian}},
journal = {Journal of Physics A: Mathematical and Theoretical}
}

@article{izvekov2025,
  title = {{Mori}-{Zwanzig} projection operator formalism: Generalized {Langevin} equation dynamics of a classical system perturbed by an external generalized potential and far from equilibrium},
  author = {Izvekov, Sergei},
  journal = {Phys. Rev. E},
  volume = {111},
  issue = {3},
  pages = {034130},
  numpages = {12},
  year = {2025},
  month = {Mar},
  publisher = {American Physical Society},
  doi = {10.1103/PhysRevE.111.034130},
  url = {https://link.aps.org/doi/10.1103/PhysRevE.111.034130}
}

@article{kawai2011,
    author = {Kawai, Shinnosuke and Komatsuzaki, Tamiki},
    title = {Derivation of the generalized {Langevin} equation in nonstationary environments},
    journal = {The Journal of Chemical Physics},
    volume = {134},
    number = {11},
    pages = {114523},
    year = {2011},
    month = {03},
    issn = {0021-9606},
    doi = {10.1063/1.3561065},
    url = {https://doi.org/10.1063/1.3561065},
    eprint = {https://pubs.aip.org/aip/jcp/article-pdf/doi/10.1063/1.3561065/13524774/114523_1_online.pdf},
}

@article{netz2024,
  title = {Derivation of the nonequilibrium generalized {Langevin} equation from a time-dependent many-body {Hamiltonian}},
  author = {Netz, Roland R.},
  journal = {Phys. Rev. E},
  volume = {110},
  issue = {1},
  pages = {014123},
  numpages = {22},
  year = {2024},
  month = {Jul},
  publisher = {American Physical Society},
  doi = {10.1103/PhysRevE.110.014123},
  url = {https://link.aps.org/doi/10.1103/PhysRevE.110.014123}
}

@article{widder_addendum,
doi = {10.1088/1751-8121/ae7168},
url = {https://doi.org/10.1088/1751-8121/ae7168},
year = {2026},
month = {jun},
publisher = {IOP Publishing},
volume = {59},
number = {22},
pages = {229401},
author = {Widder, Christoph and Zimmer, Johannes and Schilling, Tanja},
title = {ADDENDUM: On the generalized {Langevin} equation and the {Mori} projection operator technique (2025 {J. Phys. A: Math. Theor}. 58 405001)},
journal = {Journal of Physics A: Mathematical and Theoretical}
}

@misc{widder2026,
      title={Generalised {Langevin} Dynamics: Significance and Limitations of the Projection Operator Formalism}, 
      author={Christoph Widder and Tanja Schilling},
      year={2026},
      eprint={2604.20453},
      archivePrefix={arXiv},
      primaryClass={math-ph},
      url={https://arxiv.org/abs/2604.20453}, 
}

@book{cohn,
  title={Measure Theory},
  author={Cohn, D. L.},
  doi={https://doi.org/10.1007/978-1-4614-6956-8},
  url={https://link.springer.com/book/10.1007/978-1-4614-6956-8},
  year={2013},
  publisher={Birkhäuser New York, NY}
}

@book{folland,
  author = {Folland, Gerald B.},
  publisher = {John Whiley and Sons},
  title = {Real Analysis - Modern Techniques and Their Applications - Second Edition},
  year = 1999
}

@book{vanNeerven1992,
author="van Neerven, Jan",
title="The Adjoint of a Semigroup of Linear Operators",
year="1992",
publisher="Springer Berlin Heidelberg",
address="Berlin, Heidelberg",
isbn="978-3-540-47497-5",
doi="10.1007/BFb0085008",
url="https://doi.org/10.1007/BFb0085008"
}

@book{rudin1974,
  title={Functional Analysis},
  author={Rudin, W.},
  edition={TMH EDITION},
  year={1974},
  publisher={McGraw-Hill}
}

@book{evans,
  title={Partial Differential Equations},
  author={Evans, L.C.},
  isbn={9780821849743},
  lccn={2009044716},
  series={Graduate studies in mathematics},
  year={2010},
  publisher={American Mathematical Society}
}

@book{Brezis2011,
author="Brezis, Haim",
title="Functional Analysis, {Sobolev} Spaces and Partial Differential Equations",
year="2011",
publisher="Springer New York",
address="New York, NY",
isbn="978-0-387-70914-7",
doi="10.1007/978-0-387-70914-7",
url="https://doi.org/10.1007/978-0-387-70914-7"
}

@article{davies,
 ISSN = {03794024, 18417744},
 URL = {http://www.jstor.org/stable/24715654},
 author = {E. B. Davies},
 journal = {Journal of Operator Theory},
 number = {2},
 pages = {225--249},
 publisher = {Theta Foundation},
 title = {SEMIGROUP GROWTH BOUNDS},
 urldate = {2024-08-27},
 volume = {53},
 year = {2005}
}

@article{kawasaki73,
    doi = {10.1088/0305-4470/6/9/004},
    url = {https://dx.doi.org/10.1088/0305-4470/6/9/004},
    year = {1973},
    month = {sep},
    publisher = {},
    volume = {6},
    number = {9},
    pages = {1289},
    author = {K Kawasaki},
    title = {Simple derivations of generalized linear and nonlinear {Langevin} equations},
    journal = {Journal of Physics A: Mathematical, Nuclear and General}
}

@article{zhu18,
    author = {Zhu, Yuanran and Dominy, Jason M. and Venturi, Daniele},
    title = "{On the estimation of the {Mori-Zwanzig} memory integral}",
    journal = {Journal of Mathematical Physics},
    volume = {59},
    number = {10},
    pages = {103501},
    year = {2018},
    month = {09},
    issn = {0022-2488},
    doi = {10.1063/1.5003467},
    url = {https://doi.org/10.1063/1.5003467}
}

@article{berkowitz1983,
    author = {Berkowitz, M. and Morgan, J. D. and McCammon, J. Andrew},
    title = {Generalized {Langevin} dynamics simulations with arbitrary time‐dependent memory kernels},
    journal = {The Journal of Chemical Physics},
    volume = {78},
    number = {6},
    pages = {3256-3261},
    year = {1983},
    month = {03},
    issn = {0021-9606},
    doi = {10.1063/1.445244},
    url = {https://doi.org/10.1063/1.445244}
}

@article{widder2025,
doi = {10.1088/1751-8121/ae02cc},
url = {https://doi.org/10.1088/1751-8121/ae02cc},
year = {2025},
month = {sep},
publisher = {IOP Publishing},
volume = {58},
number = {40},
pages = {405001},
author = {Widder, Christoph and Zimmer, Johannes and Schilling, Tanja},
title = {On the generalized {Langevin} equation and the {Mori} projection operator technique},
journal = {Journal of Physics A: Mathematical and Theoretical}
}

@article{friedrichs1954,
author = {Friedrichs, K. O.},
title = {Symmetric hyperbolic linear differential equations},
journal = {Communications on Pure and Applied Mathematics},
volume = {7},
number = {2},
pages = {345-392},
doi = {https://doi.org/10.1002/cpa.3160070206},
url = {https://onlinelibrary.wiley.com/doi/abs/10.1002/cpa.3160070206},
eprint = {https://onlinelibrary.wiley.com/doi/pdf/10.1002/cpa.3160070206},
year = {1954}
}

@book{reed1980,
  title={I: Functional Analysis},
  author={Reed, M. and Simon, B.},
  series={Methods of Modern Mathematical Physics},
  year={1980},
  publisher={Academic Press},
  edition={Revised and Enlarged Edition}
}

@Article{givon2005,
author={Givon, Dror
and Kupferman, Raz
and Hald, Ole H.},
title={Existence proof for orthogonal dynamics and the {Mori-Zwanzig} formalism},
journal={Israel Journal of Mathematics},
year={2005},
month={Dec},
day={01},
volume={145},
number={1},
pages={221-241},
issn={1565-8511},
doi={10.1007/BF02786691},
url={https://doi.org/10.1007/BF02786691}
}

@book{engel,
  title={One-Parameter Semigroups for Linear Evolution Equations},
  author={Engel, Klaus-Jochen and Nagel, Rainer},
  series={Graduate Texts in Mathematics},
  year={2000},
  publisher={Springer New York},
  doi={https://doi.org/10.1007/b97696},
  note = {eBook 06 April 2006},
  eBook-ISBN= {978-0-387-22642-2}
}

@book{grabert2006projection,
  title={Projection operator techniques in nonequilibrium statistical mechanics},
  author={Grabert, Hermann},
  volume={95},
  year={2006},
  publisher={Springer}
}

@article{zwanzig1961,
    author = {Zwanzig, Robert},
    title = {Memory Effects in Irreversible Thermodynamics},
    journal = {Phys. Rev.},
    volume = {124},
    issue = {4},
    pages = {983--992},
    numpages = {0},
    year = {1961},
    month = {Nov},
    publisher = {American Physical Society},
    doi = {10.1103/PhysRev.124.983},
    url = {https://link.aps.org/doi/10.1103/PhysRev.124.983}
}

@article{mori1965transport,
  title={Transport, collective motion, and {Brownian} motion},
  author={Mori, Hazime},
  journal={Progress of theoretical physics},
  volume={33},
  number={3},
  pages={423--455},
  year={1965},
  publisher={Oxford University Press}
}

@book{zwanzig2001nonequilibrium,
  title={Nonequilibrium Statistical Mechanics},
  author= {Zwanzig, Robert},
  year={2001},
  publisher={Oxford University Press}
}

@book{snook2006langevin,
  title={The {Langevin} and generalised {Langevin }approach to the dynamics of atomic, polymeric and colloidal systems},
  author={Snook, Ian},
  year={2006},
  publisher={Elsevier}
}

@article{schilling2022coarse,
  title={Coarse-grained modelling out of equilibrium},
  author={Schilling, Tanja},
  journal={Physics Reports},
  volume={972},
  pages={1--45},
  year={2022},
  publisher={Elsevier}
}

@book{gotze2009complex,
  title={Complex dynamics of glass-forming liquids: A mode-coupling theory},
  author={G{\"o}tze, Wolfgang},
  volume={143},
  year={2009},
  publisher={Oxford University Press}
}

@book{breuer2002theory,
  title={The theory of open quantum systems},
  author={Breuer, Heinz-Peter and Petruccione, Francesco},
  year={2002},
  publisher={OUP Oxford}
}

@article{espanol2009derivation,
  title={Derivation of dynamical density functional theory using the projection operator technique},
  author={Espanol, Pep and L{\"o}wen, Hartmut},
  journal={The Journal of chemical physics},
  volume={131},
  number={24},
  year={2009},
  publisher={AIP Publishing}
}

@article{seegebrecht2025concept,
  title={The Concept of Minimal Dissipation and the Identification of Work in Autonomous Systems: a View from Classical Statistical Physics},
  author={Seegebrecht, Anja and Schilling, Tanja},
  journal={Journal of Statistical Physics},
  volume={192},
  number={10},
  pages={133},
  year={2025},
  publisher={Springer}
}

@article{klippenstein2021introducing,
  title={Introducing memory in coarse-grained molecular simulations},
  author={Klippenstein, Viktor and Tripathy, Madhusmita and Jung, Gerhard and Schmid, Friederike and Van Der Vegt, Nico FA},
  journal={The Journal of Physical Chemistry B},
  volume={125},
  number={19},
  pages={4931--4954},
  year={2021},
  publisher={ACS Publications}
}

\end{document}